\documentclass{article}
\usepackage[left=1in, top=1in, right=1in, bottom=1in]{geometry}
\usepackage{amsmath}
\usepackage{amssymb}
\usepackage{amsthm}
\usepackage{paralist}
\usepackage{algorithmicx}
\usepackage[noend]{algpseudocode}
\usepackage[linesnumbered,lined,boxed,commentsnumbered]{algorithm2e}
\usepackage{color}
\usepackage{bbm}

\usepackage{subcaption}
\usepackage{array}
\usepackage{multirow}

\newtheorem{lemma}{Lemma}
\newtheorem*{lemma*}{Lemma}
\newtheorem{claim}{Claim}
\newtheorem*{claim*}{Claim}
\newtheorem{theorem}{Theorem}
\newtheorem*{theorem*}{Theorem}
\newtheorem{fact}{Fact}
\newtheorem*{fact*}{Fact}

\newtheorem{definition}{Definition}
\newtheorem*{definition*}{Definition}

\newtheorem*{corollary*}{Corollary}

\newcommand{\AutoAdjust}[3]{\mathchoice{ \left #1 #2  \right #3}{#1 #2 #3}{#1 #2 #3}{#1 #2 #3} }
\newcommand{\Xcomment}[1]{{}}

\newcommand{\InParentheses}[1]{\AutoAdjust{(}{#1}{)}}
\newcommand{\InBrackets}[1]{\AutoAdjust{[}{#1}{]}}% {\left[{#1}\right]}
\newcommand{\Ex}[2][]{\operatorname{\mathbf E}_{#1}\InBrackets{#2\vphantom{E_{F}}}}

\newcommand{\Prx}[2][]{\operatorname{\mathbf{Pr}}_{#1}\InBrackets{#2}}
\newcommand{\Prlong}[2][]{\operatornamewithlimits{\mathbf Pr}\limits_{#1}\InBrackets{#2\vphantom{\operatornamewithlimits{\mathbf Pr}\limits_{#1}}}}
\def\prob{\Prx}
\def\expect{\Ex}

\newcommand{\eqdef}{\stackrel{\textrm{def}}{=}}

\newcommand{\be}{\begin{equation}}
\newcommand{\ee}{\end{equation}}
\newcommand{\bee}{\begin{equation*}}
\newcommand{\eee}{\end{equation*}}

\newcommand{\eps}{\varepsilon}
\newcommand{\vect}[1]{\ensuremath{\vec{#1}}}
\newcommand{\R}{\mathbb{R}}
\newcommand{\N}{\mathbb{N}}

\newcommand{\upround}[1]{\left\lceil #1 \right\rceil}
\newcommand{\dround}[1]{\left\lfloor #1 \right\rfloor}

\def \reals {{\mathbb R}}

\newcommand{\stopt}[1][t]{\textsf{STOP}_{#1}}

\newcommand{\distone}{F}

\newcommand{\oprice}{\price^o}

\newcommand{\mech}[0]{\mathcal{M}}

\newcommand{\rev}[0]{\textsf{Rev}}

\newcommand{\myopic}[1]{\textsf{Greedy}\left(#1\right)}

\newcommand{\myopicname}{\textsf{Greedy}}

\newcommand{\spm}{\textsf{SPM}}
\newcommand{\uspm}{\textsf{U-SPM}}

\newcommand{\opt}{\textsf{Opt}}

\newcommand{\upricing}{\textsf{Uprice}}
\newcommand{\upricingmech}[1]{\textsf{Uprice}\InParentheses{#1}}

\newcommand{\baitfamily}{\mathcal{B}}
\newcommand{\Menu}{\textsf{M}}
\newcommand{\Mexp}{\textsf{M}_{\textsf{exp}}}
\newcommand{\Mbait}{\textsf{M}_{\textsf{bait}}}
\newcommand{\Mbaite}{\textsf{M}_{\textsf{bait}}^{\textsf{e}}}
\newcommand{\Mbaito}{\textsf{M}_{\textsf{bait}}^{\textsf{o}}}
\newcommand{\Menux}{\widehat{\Menu}}

\newcommand{\lutil}[1][t]{\underline{u}_{#1}}
\newcommand{\uutil}[1][t]{\bar{u}_{#1}}

\newcommand{\val}{v}
\newcommand{\vals}{\vect{\val}}
\newcommand{\valv}{\vect{\val}}

\newcommand{\vali}[1][i]{{\val_{#1}}}

\newcommand{\util}{u}

\newcommand{\utili}[1][i]{\util_{#1}}

\newcommand{\utilx}{\widehat{\util}}
\newcommand{\utilb}{u_b}

\newcommand{\price}{p}

\newcommand{\pricev}{\vect{\price}}
\newcommand{\pricei}[1][i]{{\price_{#1}}}

\newcommand{\eT}{T_{e}}

\newcommand{\eventi}[1][i]{\mathbf{E_{#1}}}
\newcommand{\expensive}{\textsf{EXP}}

\title{Monopoly pricing with buyer search}
\author{
	Nick Gravin \\
	Shanghai University of Finance and Economics\\
	\tt{nikolai@mail.shufe.edu.cn}
	\and
	Zhihao Gavin Tang \\
	University of Hong Kong\\
	\tt{zhtang@cs.hku.hk}
}
\date{}
\begin{document}

\maketitle

\begin{abstract}
In many shopping scenarios, e.g., in online shopping, customers have a large menu of options to choose from.
However, most of the buyers do not browse all the options and make decision after considering only a small 
part of the menu. To study such buyer's behavior we consider the standard Bayesian monopoly problem for a unit-demand buyer, where 
the monopolist displays the menu dynamically page after a page to the buyer. 
The seller aims to maximize the expected revenue over the distribution of buyer's values which we assume are i.i.d. 
The buyer incurs a fixed cost for browsing through one menu page 
and would stop if that cost exceeds the increase in her utility.
We observe that the optimal posted price mechanism in our dynamic setting may have quite different structure than in the classic 
static scenario. We find a (relatively) simple and approximately optimal mechanism, that uses part of the items as a ``bait'' to 
keep the buyer interested for multiple rounds with low prices, while at the same time showing many other expensive items.
\end{abstract}

\section{Introduction}
\label{sec:intro}
% Monopolist problem. (approximation + theoretical motivation) Selling multiple goods to a customer is fundamental problem.
The monopolist problem of selling multiple goods to a single buyer is a fundamental problem in mechanism design. 
In this situation any incentive compatible interaction between the monopolist seller and a single buyer can be described as a menu of 
possible allocations and payments that the seller offers to the buyer to choose from.
Despite extensive studies, this multidimensional mechanism design problem is not very well understood in contrast 
to the Myerson's optimal auction for the single item case~\cite{myerson1981optimal}. In special cases when the 
optimal mechanisms are known, these mechanisms often exhibit irregular and complex behavior. 
For example, the revenue of the optimal auction may be non-monotone~\cite{hart2015maximal}\footnote{Optimal auction may have higher expected 
revenue for the values with stochastically dominated prior distribution.}, or
the optimal auction must offer a menu of randomized outcomes~\cite{HartN13,daskalakis2014complexity}, i.e., lotteries.  
%\zhihao{which is this work? Should it be `these works'?}\nick{``work'' is generally uncountable, unless it refers to works of art or literature...}
Another problem observed in that line of work was that the optimal auction may have a 
menu of arbitrary large (or even infinite) size even in a simple setting with a unit-demand buyer with independent values.

%We would like to highlight that those optimal complex solutions are almost never observed
%in practice but, instead, we witness prevalence of simple item pricing mechanisms in our daily life.
%
% Menu complexity - informational cost. We are operating in a regime where sublinear menu size has a non negligible cost vs. 
%previous polynomial/subexponential bounds.
%\paragraph{Optimal Auctions \& Menu Complexity.}
%Optimal auctions in simple settings such as additive or unit-demand monopoly problems are often unrealistically complex. 
%For example, solution in the case of monopoly for unit-demand i.i.d. buyer may involve lottery pricing, or the menu size could be 
%large or even infinite~\cite{HartN13,daskalakis2013mechanism, daskalakis2014complexity}.

%Menu complexity
%Simplicity and menu complexity of auctions
%- simplicity for buyer vs. seller
%- argue that theory leads in some way or another to a behavioral model.   
These critiques have motivated~\cite{HartN13} to propose the {\em menu size} as a measure of auction complexity 
and study its impact on the revenue in the classic monopoly problem. Since~\cite{HartN13} there has been more recent work giving various upper 
and lower bounds on the menu complexity of optimal or approximately optimal auctions in different scenarios~\cite{DughmiHN14,WangT14,Yannai17,Yannai18}. However, this work gives
bounds on the menu sizes as a restriction on the seller, which actually does not make much sense from the seller's point of view. Indeed, most of the imaginable 
sellers would not hesitate to use a complex mechanism provided that it will generate more revenue than a simple one. Moreover,  many sellers are capable
and willing to do sufficient research and do find many ways to maximize their revenue. Hence, a more accurate explanation of the prevalence of simple mechanisms
would be that simplicity is a property desired by the buyer, but not the seller. This line of thought unavoidably implies that one has 
to make certain behavioral assumptions on the buyer's interaction with the seller. In this work we propose a new simple 
theoretical model that combines and rationalizes the buyer's desire for simplicity and the seller's desire to maximize his revenue.
%take the buyer's point of view and

%Examples and stories
%restaurant
%\zhihao{Shall we omit the details that the menu is in Chinese and hard to translate? Somehow this is another bad feature of the menu and we should focus on that the `length/size' is too large.}
%\nick{Well, I think it is not that distracting and without all the details the story would sound a bit illogical -- if you understand the language you can quickly search in the menu.}
Let us first illustrate with a few examples the importance of simple mechanisms, i.e., short menus, from the buyer's point of view. 
We begin with a personal story that happened to one of the authors of this paper. 
This author had once a visitor who is a foodie, partial to Szechuan food. They went for dinner to a high end Chinese 
restaurant, which is famous for its variety of dishes and provides very detailed menu. In that case the menu contained a list of more than $100$ 
items each supplemented with a note describing ingredients and a cooking method. 
Everything on the menu was in Chinese, since many ingredients don't have proper translation into English. 
The visitor could not understand any Chinese and asked a waiter to help him out to choose a dish.
The waiter began reading through the menu, item by item, starting with some exotic dishes, which did not seem appealing to the visitor. 
He had to ask the waiter to stop and after some argument forced the waiter to simply suggest three best dishes of the waiter's choice.
Both of us followed the waiter's advice and ordered one out of three recommended dishes. 
Eventually, everyone greatly enjoyed the dinner, but during most of the meal the discussion was turning around the issues of complex menus and 
the necessity for the restaurants to give a short recommendation list on the front page of their menu. 
%%We began reading through the menu, item by item, starting with some exotic dishes, which did not seem appealing.
%%The visitor called for a waiter and after some argument forced the waiter to simply suggest three best dishes of the waiter's choice.
%eventually, they ordered food and everyone enjoyed the dinner, but during most of the meal the discussion was turning around the issues of complex menus and 
%the necessity for the restaurants to give a short recommendation list on the front page of their menu.

%eBay
Our next example is related to the online shopping scenario that should be familiar to most of our readers. Imagine, for example, that a buyer wants to buy a computer 
mouse on one of the online shopping platforms (e.g., on Amazon). The buyer starts her search by seeing a single web page, say, with $18$ different entries out of more than
$1000$ available items on the platform. She may browse a few more pages, but she is likely to choose a mouse from among the first three pages, even though there might be over 
$50$ pages that she did not check with potentially better or cheaper products than the best one she has found.

These examples lead us to study the interaction between the buyer and the seller as
%the item pricing mechanism as 
a dynamic procedure in which items on the menu are shown sequentially page-after-a-page, 
rather than as a static single-list menu. It is natural to assume that the buyer's decision to continue browsing depends on the buyer's impression of 
the previous pages. For example, the buyer may see a good item at the very first page and decide to stop and buy this item after seeing another menu page with 
less attractive offers than the first one. The impatient behavior of the buyer actually has a rational explanation. Indeed, people incur non-negligible costs for 
their search efforts. Thus it is not surprising that our buyer would not want to spend an hour fruitlessly browsing $990$ out of $1000$ items before she makes
a purchase.

%model & reseach question 
In this paper we study how dynamic impatient behavior of the buyer can affect the design of pricing mechanisms. To this end we consider a simple monopoly problem, where 
the monopolist gradually displays a menu to a single unit-demand buyer. Naturally, we restrict our attention only to item-pricing mechanisms\footnote{Indeed, a lottery is 
incomparably harder to evaluate for the buyer than an item price: (i) every item in the support of the lottery's distribution requires separate consideration; 
(ii) humans are usually risk-averse and not very good at reasoning about probabilistic outcomes; (iii) there is a big issue of trust in the randomness of the lottery.}.
\paragraph{Model for Impatient Buyer.} We consider, a typical Bayesian mechanism design setting, where the seller has 
an inventory of $m$ items and tries to sell them to a unit-demand buyer. The seller has a prior knowledge of the buyer's value distribution for the goods 
which are assumed to be i.i.d. The seller's objective is to maximize the revenue. The interaction between the seller and the buyer goes in consecutive 
stages. At each stage $t\in\N$, the buyer looks at a menu page $M(t)$ of a certain size $k$, then she decides whether to continue, i.e., to see the next menu page, or to 
stop browsing the menu and take the best offer from the seller she has seen so far. 
In general, the decision of the buyer whether to stop $\stopt=1$, or to continue $\stopt=0$ at stage $t$ depends on the parts of the menu shown to the buyer and the buyer's valuations $\valv$ for the relevant items on the menu and it may be a random variable. 
%It may be a random variable. E.g., it can work as follows: the buyer can stop at each stage with a certain constant probability 
%(i) irrespective of her valuation profile, or (ii) when she does not like any of the items on a menu page $M(t)$, or (iii) when the increase of buyer's utility after seeing page $M(t)$
%is not larger than a given threshold amount. 
%The actual decision $\stopt$ can be a probabilistic combination of all three previous rules and maybe some other factors\footnote{It is an 
%interesting econometrics question to understand real consumer's behavior in different situations such as searching on e-commerce platforms, or returning to retail stores. 
%%Unfortunately, we are not aware of any work on this topic in econometrics or behavioral economics literature.
%}. 
%\nick{todo: our model choices: searching costs \& exploration. Relation \& difference with paper by Susan Athey}
Similar to the model of~\cite{AtheyEllison11} we assume that consumer incurs a certain fixed cost $\Delta$ when performing her search. This
cost is paid per menu page rather than per item unlike in~\cite{AtheyEllison11}\footnote{It is easy to adapt our model for the case where search costs are paid per item rather than per page by splitting the menu page cost between $k$ menu items.}. We further assume that the decision of the buyer $\stopt\in\{0,1\}$ depends on the buyer's utility increment compared to the previous stages. Namely, we assume that the buyer stops when the increase of her utility after seeing page $M(t)$ is not larger than her cost $\Delta$.

These modeling choices\footnote{We would like to mention that without identical assumption on the prior distribution the revenue maximization problem becomes 
quite non trivial already in the static setting, i.e., where the seller has to select up to $k$ out of $m$ different items to put on a single menu page. 
We chose not to study the case of non identical distributions more traditionally considered in the literature, because it would have been distracting for the interesting 
structural insights and simplicity of the model that could be seen more clearly in the i.i.d. case.} allow us to capture large uncertainty and exploration nature of the buyer's 
behavior in the relevant settings. In this sense the size $k$ of a menu page can be viewed as a tolerance parameter for the buyer's willingness to explore.
%whether the increase of the buyer utility after seeing page $M(t)$ is not larger than a given threshold $\Delta$.
%
%In this work we focus on a very simple and natural deterministic rule for the buyer's stopping behavior. Namely, we assume that the buyer stops when the increase of her utility after seeing page $M(t)$ is not larger than a given threshold $\Delta$.
%the previous menu shown to the buyer and the buyer's valuations $\valv$ for the relevant items on the menu and it may be a random variable. 
On the other hand, our model choice is motivated by the behavioral economics model of rational inattention~\cite{Sims03}. In this theory, agent is assumed to be rational, but having limited 
and rather scarce amount of attention that she can spend to access information to her benefit. Then the rationale of the buyer is that she would want to see more of the 
menu only if the improvement to her utility is worth her time, effort, or attention. Conversely, she does not want to continue if her previous experience of browsing the menu 
$M(t)$ did not improve her utility by a certain threshold amount $\Delta$. 
We believe that this very simple and reasonable rule, on the one hand, captures some real features of the consumer's behavior, and, on the other hand, offers an 
easy-to-state and interesting mechanism design problem.

\subsection{Example}
\label{sec:example}
To illustrate the problem let us consider the following simple scenario, in which the buyer's values are distributed i.i.d. according to 
$\distone:\prob[v\sim\distone]{v=10}=0.9, \prob[v\sim\distone]{v=100}=0.1$; the menu size $k=2$; search cost $\Delta=1$; and item supply is unlimited, i.e., $m=\infty$. 
Let us first consider a mechanism that posts uniform prices on the items on each menu page $\Menu(t)$, where $t\in\N$. Indeed, it seems reasonable to 
use uniform prices because of the symmetry between all items in the supply. It is not very difficult to calculate the optimal menu (see table~\ref{tab:uniform_menu}), 
which has only $2$ pages. 

\begin{table}[h]
\begin{subtable}[c]{0.4\textwidth}
\centering
\begin{tabular}{| c | c | c |}
 \multicolumn{1}{c}{$\Menu(1)$} & \multicolumn{1}{c}{~} & \multicolumn{1}{c}{$\Menu(2)$} \\
\cline{1-1}\cline{3-3}
$9$\$ &~& $98.9$\$\\
\cline{1-1}\cline{3-3}
$9$\$ &~& $98.9$\$\\
\cline{1-1}\cline{3-3}
\end{tabular}
\caption{Optimal menu sequence with uniform prices per page. Expected revenue 22.8\$.}
\label{tab:uniform_menu}
\end{subtable}
\quad
\begin{subtable}[c]{0.55\textwidth}
\centering
\begin{tabular}{| c | c | c | c | c | c | c |}
\multicolumn{1}{c}{$\Menu(1)$} & \multicolumn{1}{c}{~} & \multicolumn{1}{c}{$\Menu(2)$} & \multicolumn{1}{c}{~} & \multicolumn{1}{c}{$\Menu(3)$} 
& \multicolumn{1}{c}{\quad} & \multicolumn{1}{c}{$\Menu(10)$} \\
\cline{1-1}\cline{3-3}\cline{5-5}\cline{7-7}
$9$\$ &~& $8$\$ &~& $7$\$ & \multirow{2}{*}{$\cdots$} & $0$\$ \\
\cline{1-1}\cline{3-3}\cline{5-5}\cline{7-7}
$97.9$\$ &~& $96.9$\$ &~& $95.9$\$ & & $88.9$\$\\
\cline{1-1}\cline{3-3}\cline{5-5}\cline{7-7}
\end{tabular}
\caption{Mechanism with non-uniform prices per menu page. The expected revenue is 38.3133\$.}
\label{tab:good_menu}
\end{subtable}
\caption{Examples of mechanisms with menu pages dynamically revealed to the buyer. All prices on the tables are given for different items.}\label{tab:1}
%\caption{Examples of mechanisms with dynamically revealed menu pages $\{\Menu(t)\}_{t=1}^{\infty}$ shown to the buyer from $t=1$ onwards.}
\end{table}

The revenue of the mechanism~\ref{tab:uniform_menu} is a guaranteed $9\$$ plus the extra surplus the seller gets, if the buyer likes one of the expensive items on the $2$nd page.
In total, the revenue of mechanism~\ref{tab:uniform_menu} is $9 + \prob{\forall i\in\Menu(1):\vali=10}\cdot\prob{\exists i\in\Menu(2):\vali=100}\cdot\left(98.9-9\right)=
9+ 0.81\cdot 0.19\cdot 89.9=22.8.$
Now consider another mechanism with a sequence of non-uniform price menus with $20$ different items described in the Table~\ref{tab:good_menu}. In turns out that 
computing the expected revenue of mechanism~\ref{tab:good_menu} is not an easy numerical exercise that involves 
calculations of state distributions after up to $10$ steps in a certain Markov chain. Instead of doing that, we provide much simpler 
approximate estimate of mechanism~\ref{tab:good_menu} revenue. Consider the event that out of $20$ items shown on all $10$ pages 
there is at least one item of high value $100.$ Then, the buyer continues browsing menu pages until she gets to see her first high-value item.
With slightly less than $0.5$ probability, this item has a high price and the remaining item on the current page and both of the items on the next menu page have low values, in which case
the buyer stops and buys this expensive item. This estimate gives us an approximate value of the mechanism~\ref{tab:good_menu} revenue of $(1-0.9^{20})\cdot 0.5\cdot 100=44$.

We note that the revenue of the mechanism~\ref{tab:good_menu} (we know that mechanism~\ref{tab:good_menu} is not optimal) is significantly larger than the revenue of the best 
uniform price mechanism~\ref{tab:uniform_menu}, and that it is not hard to modify our example to make the gap between the non-uniform and uniform price mechanisms to be arbitrarily large.
On the other note, the above example already highlights the importance of approximate (versus exact) analysis in our setting. Indeed, it is unlikely that we can find a mechanism with the 
optimal revenue, when it is already difficult to compute the revenue of a given mechanism in such a small example.

%we describe simple example with $k=2$ and $\distone=\{v=10 \text{ pr. .9}, v=100 \text{ pr. .1}\}$, $\Delta=1$, $m$ is unlimited. The mechanism is a good mechanism that uses 1 bait item and one expensive item per page. We want to emphasize that even in this simple scenario it is very non trivial task to compute the exact revenue of the mechanism. However, it is relatively easy to give a close estimate.

\subsection{Related Work}
\label{sec:related}
Inherent complexity of optimal auctions led to the study of approximately optimal simple auctions under the name of ``simple versus optimal mechanisms''~\cite{HartlineR09}.
The most related to our setting is a series of papers with nearly optimal sequential posted pricing mechanism for unit-demand buyer~\cite{ChawlaHK07,ChawlaHMS10,ChawlaMS15}.
The approximation factor is $8$ against optimal randomized mechanism and $2$ against optimal item pricing. For the same unit-demand pricing problem~\cite{CaiD11} designs a 
PTAS and quasi-PTAS for some special classes of distribution families, although these solutions are not as simple as Sequential Pricing Mechanism (SPM)
in~\cite{ChawlaHK07,ChawlaHMS10,ChawlaMS15}. Papers~\cite{ChawlaHMS10,ChawlaMS15} also extend single buyer results to BIC mechanisms for many buyers. However, even 
the multi-buyer case of SPM is closely related to an important primitive mechanism $\myopicname$ in our work, as we show in Appendix~\ref{app:sec:spm}.
Similar to our setting~\cite{CaiD11,ChawlaHK07,ChawlaHMS10,ChawlaMS15} consider Bayesian setting 
with independent values\footnote{This work considers not necessarily identical distributions.}. The difference with our work is that in our model the menu complexity 
affects utility of the buyer. On the other note, there is a large body of literature on Bayesian mechanism design, which is too big to describe here. 
For an extensive survey on the topic see~\cite{ChawlaS14}.

%+ prohphet inequality
%\paragraph{Prophet Inequality.} 
The main technique in~\cite{ChawlaHK07,ChawlaHMS10,ChawlaMS15} is based on the prophet inequality,
which first appeared in~\cite{SamuelCahn84} for a simple gambler's problem and was first adopted to mechanism design 
literature in~\cite{ChawlaHK07,HajiaghayiKS07}. 
For more recent results on prophet inequality see~\cite{AbolhassaniEEHK17,AzarKW14}.
Other application of prophet inequalities include, e.g., optimization on matroids~\cite{KleinbergW12}, and
Bayesian combinatorial auctions~\cite{FeldmanGL15}. 
Closer to our problem, \cite{DuttingFK16} explicitly study 
the revenue gap between discriminatory and anonymous sequential posted pricing. It was shown 
in~\cite{ChawlaHK07, ChawlaHMS10, DuttingFK16} that the revenue gap between optimal Sequential Pricing Mechanism (SPM) 
and optimal uniform-SPM is $2$. Some of these techniques and general philosophy of approximation with simple pricing schemes 
is adopted in our work.
%As a part of our argument 
%we also need to prove similar gap, but for a more general pricing problem with impatient buyer. The latter result 
%might be of independent interest. 

%Behavioral Economics
%\paragraph{Behavioral Economics.} 
We assume certain simple and reasonable impatient behavior of the buyer that fits into 
the behavioral economics model of rational inattention~\cite{Sims03}. Another line of work on behavioral economics 
models~\cite{KleinbergO14,TangTWXX17,AlbersK17,GravinILP16,KleinbergOR17} in algorithmic game theory literature concerns time-inconsistent planing.

%\cite{AtheyEllison11}
Our model for the buyer search behavior in some ways is similar to \cite{AtheyEllison11}. This paper belongs to a rather 
rich literature in economics studying how consumer search, i.e., situations where consumers incur certain costs to acquire information about products. 
%There is a rich literature in economics studying how consumer search, i.e., situations where consumers incur certain costs to acquire information about products. 
We refer interested readers to~\cite{ratchford2009consumer} for a survey. 
Those models usually study market outcomes and therefore have to assume simpler consumer behavior than in our case. 
In contrast, we assume that the buyer has perfect knowledge about her private valuation and focus on the dynamic 
interaction with the mechanism designer.

% Dynamic mechanism design
%\paragraph{Dynamic Mechanism Design.} 
Finally, our model is closely related to dynamic mechanism design, as the buyer in our setting decides dynamically
when to leave. Computer science literature on dynamic mechanism design, e.g.,~\cite{PapadimitriouPP16,AshlagiDH16,BalseiroML17}, usually consider
multi-round interaction scenario with certain ex-post IC, or IR constraints. In economics the literature on dynamic mechanism design is rather large, see~\cite{BergemannM10} 
for a detailed survey. The closest to our paper are the papers~\cite{GershkovM10},~\cite{Said09}. In \cite{GershkovM10} 
heterogeneous durable goods are dynamically allocated to randomly arriving impatient buyers, whereas~\cite{Said09} examines the allocation of a sequence of indivisible and perishable goods to a 
dynamic population of patient unit-demand buyers with i.i.d. private valuations.
Our setting is simpler than~\cite{GershkovM10} and \cite{Said09} with only one buyer and fixed arrival time. On the other hand, we consider more complex dynamic model for the buyer
that interpolates between patient and myopic behaviors and captures the explore-or-exploit tradeoff faced by her.
%The novelty of our approach lies in the new model of buyer's behavior that better captures the explore-or-exploit tradeoff faced by the buyer.

\subsection{Our Results and Techniques}
\label{sec:results}
As we already observed in Section~\ref{sec:example} some approximation analysis may be necessary in our setting. As it turns out the problem of 
finding approximately optimal mechanism is tractable and we can find a relatively simple mechanism that achieves a constant approximation to the 
optimal revenue (see Theorem~\ref{thm:one_time}). The high-level principles behind our solution are similar to those used in mechanism~\ref{tab:good_menu} from Section~\ref{sec:example}.
The mechanism~\ref{tab:good_menu} and more generally a family of simple-to-analyze mechanisms with approximately optimal revenue can be described as 
follows.

\begin{definition*}[Bait Mechanisms $\baitfamily$]
All items on the menu can be divided into the two categories.  
%A mechanism $\mech$ is a \emph{bait mechanism} if all menu items (i.e. item-price pairs) can be classified into two different types, namely \emph{bait items} and \emph{expensive items} so that:
\begin{description}
	\item[Bait items.] Normally, these are the cheap items with a high probability to be liked by the buyer. Bait items encourage the buyer to continue browsing the menu for multiple rounds. 
	For each menu page $\Menu(t)$, at most $2$ different prices are used for the bait items.
	\item[Expensive items.] These items generate revenue. Normally, the buyer would not like any of the expensive items over the bait items
	at any given menu page. However, in a long run the buyer still might find an expensive item more preferable over all previous items.
	For each menu page $\Menu(t)$, expensive items have the same or similar price\footnote{If the distribution $\distone$ satisfies certain reasonable condition (formally specified later) we can 
	use uniform price for the expensive items. For arbitrary distributions the seller may need to use variable prices. Still, all the prices on the expensive items will be 
	within a constant factor from each other and can be computed efficiently.}.
\end{description}
\end{definition*}

Intuitively, bait items play a role of attracting the buyer to browse through many items until she sees an expensive item that she likes. After that, with certain probability, she stops at the next page and takes that expensive item. We are only interested in the revenue extracted when this event occurs, i.e., an expensive item is sold. We also note that a single menu of size $k$ also 
belongs to the family $\baitfamily$, as it uses only expensive items to extract revenue in a single round.

\paragraph{Techniques.} At a high level, our proof proceeds by iteratively simplifying a given mechanism to 
the desired format so that the seller suffers only a constant factor loss to the revenue. We start with the optimal pricing scheme and separate
prices on every menu page into bait and expensive parts. The analysis and simplification of the expensive items is relatively easy,
as we can use standard techniques to approximate the revenue with uniform pricing. The analysis of the bait part, however, is much trickier,
since we need to care about concentration of the buyer's utility at every time step and find a balance between the bracket of the utility increment 
and our confidence estimate of the utility (as a random variable) to fit into this bracket. One of the key ingredients in our proof 
is Lemma~\ref{lem:utility_control} that allows us to separate the analysis of the bait items to individual menu pages.
In mathematical terms, Lemma~\ref{lem:utility_control} states that one can find an almost disjoint partition of confidence intervals 
for independent random variables provided that the sequence of these variables is monotonically increasing. Another important step 
in our proof is Lemma~\ref{lem:baits_2_prices} that shows that two different prices suffice to do approximate bi-criteria optimization for 
(i) the buyer's utility increment, and (ii) the confidence bound for the utility bracket. With the help of these two lemmas, we 
show that a bait mechanism can achieve similar level of control over the buyer's utility as the optimal mechanism.
Quite surprisingly, it is not that simple to add the expensive items back to the simplified bait menu pages. The reason is that the expensive 
items may interfere with the effect of bait items in such a way that the buyer would stop because of an expensive item, but eventually
she still prefers to buy a cheap bait item. Because of this reason we have to tune the prices of the expensive items 
(although, only within a constant factor) to avoid the latter problem.

%
%\nick{todo: write techniques.}
%At a high level, our proof breaks down into the following steps:
%\begin{enumerate}
%%\item We first establish an upper bound of the optimal revenue. In particular, we carefully divide all items on the optimal menu into two categories, corresponding to our bait and 
%%expensive items respectively. 
%%We consider the revenue extracted from a myopic buyer by showing expensive items to her as an estimation of the optimal revenue.
%\item The second step is to obtain a relatively simple bait menu so that the buyer is willing to see the whole menu with constant probability.
%We start by removing the expensive items from the optimal menu, after which, the buyer is still willing to see the whole menu with high probability. By the behavior assumption of our model, it is equivalent to say that her utility monotonically increases with high probability. 
%Our key Lemma~\ref{lem:utility_control} enables us to quantify this monotone structure and obtain disjoint confidence intervals of the buyer's utilities. Finally in Lemma~\ref{lem:baits_2_prices}, we show that for each menu page, two different prices suffice to have the same utility control.
%\item Finally, we add back expensive items into the simplified bait menu and prove that with decent probability, the buyer would like to search the menu until she likes one of the expensive items. The main technical difficulty is that expensive items may interfere with the effect of bait items, so that the buyer would stop in an earlier stage but only pays for a bait item.
%\end{enumerate}

\section{Preliminaries}
\label{sec:prelim}
The monopolist sells $m$ items to the impatient unit-demand buyer, i.e., the buyer is only interested in buying at most one item. 
The buyer's values for the items are drawn i.i.d. from a given distribution $\distone$, which is known to the
seller. For computational reasons we assume that $\distone$ has a finite support\footnote{All our non computational results extend to the case when $\distone$
is arbitrary distribution with a bounded support on $[0,\infty)$.}  on $[0,\infty)$. We write buyer's aggregate valuation profile as $\valv \in \reals^m_{\ge 0}$.
%The buyer's valuation profile $\valv \in \reals^m$ is drawn form the product distribution $\distone^m$, where  
%$\distone$ is a given distribution supported\footnote{We assume that $\distone$ has bounded support and no atoms (to avoid tie-breaking rules).} on $[0,\infty)$. 
%\nick{Hmm, what's about finite support for computational purposes? Also ties are fine, just break them in favorable for us way.}
%\zhihao{just say F is a given distribution with finite support? I think no one cares about tie-breaking, see footnote 8 below.}

A mechanism $\mech$ is defined as the sequence of menus $\{\Menu(t)\}_{t=1}^{\infty}$, where each $\Menu(t)$ is a list of at most $k$ item prices.
For each item price $(i,\pricei)$ on the menu, the buyer derives utility of $\utili=\vali-\pricei$, if she takes this offer.
At each stage $t$ of the mechanism $\mech$, the buyer is shown the menu page $\Menu(t)$, 
and her utility $\util(t) \eqdef \max_{(i,\pricei)\in \Menu(t)} \utili$.
To simplify notations, we define $\util(0)\eqdef 0$ at stage $0$.
In this work we study posted price mechanisms, i.e., we assume that the price of any item, once it is posted cannot be changed at a later stage. 
Equivalently, the seller is allowed to show every item to the buyer only once. Therefore, as all items are symmetric, we shall omit the items' 
identities when describing a menu page, i.e., each menu page $\Menu(t)$ will be given as a set of prices.
The buyer's decision at time $t$ of whether to stop, or to continue to the next page depends on her utility increment compared to the previous stage. 
She continues when her utility $\util(t)$ at stage $t$ increases at least by a given parameter $\Delta$ compared to the previous stage, i.e., $\util(t) \ge \util(t-1) + \Delta$.
Otherwise, the buyer stops and takes the best offer from the menu $\bigcup_{s=1}^{t}\Menu(s)$.
In other words, if $s^*$ is the stopping time, then the buyer takes the offer $(i^*, \price^*)\in\bigcup_{t=1}^{s^*}\Menu(t) : \utili[i^*]=\max_{t\le s^*}\util(t)$ and pays $\price^*$\footnote{In case of a tie we assume that the buyer takes the highest price offer.}.
We denote by $\rev(\mech) \eqdef \expect[\valv\sim\distone^m]{\rev(\mech(\valv))}$ the total expected revenue of a given mechanism $\mech$ and by
%$\rev(\mech(\valv))$ the revenue of the seller when the buyer has valuation profile $\valv$ and denote by 
$\opt \eqdef \max_{\mech}\rev(\mech)$ the revenue of the optimal mechanism.

\paragraph{Greedy Buyer.} To facilitate the analysis, we consider a very simple greedy price taking behavior of the buyer, which will allow 
us to give a simple upper bound on the revenue of any mechanism. To obtain this bound, we order the prices in a finite menu $M=\{\price_t\}_{t=1}^n$ in 
a decreasing order and show them one by one to the buyer. The buyer takes the first item that gives her non-negative utility.
Then the revenue of $\myopic{\valv}=\max\left\{\price \left | \pricei\in M: \vali \ge \pricei \right .\right\}$.

\begin{definition}[Greedy] $\myopic{M}\eqdef\Ex[\valv]{\myopic{\valv}}$ for a given menu $M$.
%Define $\myopic{M | \valv_J}\eqdef \Ex[\valv_{\text{-}J}]{\myopic{M | \valv_J, \valv_{\text{-}J}}}$ for a fixed partial valuation $\valv_{J}$ 
\end{definition}

We slightly abuse the previous notation and denote by $\myopic{n}$ the maximum revenue of a menu of size $n$ that can be extracted from a greedy buyer, 
i.e., $\myopic{n} \eqdef \max_{|M|\le n}\myopic{M}$.

\paragraph{Uniform Pricing.}
We define $\upricing(\ell,\price)$ to be the expected revenue of offering a uniform item price $p$ in the menu of size $\ell$. Specifically,
$\upricing(\ell, \price) = (1-F(\price)^\ell) \cdot \price.$ We again slightly abuse this notation and use $\upricing(\ell)$ to denote the 
optimal revenue achieved by posting the optimal uniform price in the menu of size $\ell$.
%Formally, we have the following definition.

\begin{definition}[U-Pricing]	$\upricing(\ell) \eqdef \max_{\price} \upricing(\ell,\price)$.
\end{definition}

We establish the following simple property of the uniform pricing mechanism.
\begin{claim}
	\label{cl:umenu_half}
$\upricing(c \cdot \ell) \le c \cdot \upricing(\ell)$ for all $c,\ell\in\N$.
\end{claim}
\begin{proof}
Let $\price$ be the optimal price of $\upricing(c \cdot \ell)$. Consider posting any given price $p$ over $\ell$ items,
%\begin{align*}
$\upricing(\ell) \ge \upricing(\ell, \price) = (1-F(\price)^{\ell}) \cdot \price \ge \frac{1}{c}(1-F(\price)^{c\ell}) \cdot \price = \frac{1}{c} \cdot \upricing(c \cdot \ell).\qedhere$
%\end{align*}
\end{proof}

\section{Bait Mechanisms}
\label{sec:one_time}
The main result in this section is to show that some bait mechanism from the family $\baitfamily$ gives a constant approximation to the optimal revenue. We further show
that a bait mechanism with constant approximation to the optimum can be computed in polynomial time. 
%Moreover, an $O(1)$-competitive bait mechanism can be computed efficiently.
%We first define the family of \emph{bait mechanisms}.
%\emph{$(\Delta, \eta)$-spreading distributions}
%\zhihao{The following definition is not good. It is the same as saying at most $3$ different prices on each page.}

%\paragraph{Bait Mechanisms $\baitfamily$.}%\nick{bait mechanism, baiting has alternative meaning and negative connotation}

\begin{theorem}
	\label{thm:one_time}
$\opt \le O(1) \cdot \max_{\mech \in \baitfamily} \rev(\mech)$. Approximately optimal $\mech\in\baitfamily$ can be computed in polynomial time.
%An $O(1)$-competitive mechanism $\mech \in \mathcal{B}$ can be computed efficiently.
\end{theorem}

%\zhihao{how to find such a mechanism efficiently?}

%We modify the optimal mechanism in the following way, such that revenue downgrades by at most a constant factor each step and the final mechanism falls into the simple baiting mechanisms family:
%\begin{itemize}
%	\item It suffices to consider mechanisms that survive to the last menu page with probability $\frac{9}{10}$.
%	blahblahblah
%	\item Next, we separate the menu items into baits and offers. We show that the revenue achieved can be bounded above by the revenue that showing only offers to a myopic buyer.
%	\item By Lemma~\ref{lem:uniform_offer}, the prices of offers can be made uniform.
%	\item Finally, we change the prices of baits into at most 2 different prices within each menu page.
%\end{itemize}
\begin{proof}
Let $\mech_0$ be the optimal mechanism, i.e. $\rev(\mech_0) = \opt$. Our proof strategy will be to simplify the optimal mechanism such that it has a simple structure and at the same time
it extracts a constant fraction of the optimal revenue. First, we truncate the optimal menu so that the buyer goes until the end of the menu with constant probability. 
Let $T$ be the largest page number so that the probability of surviving until time $T$ is least $\frac{11}{12}$, i.e. the buyer sees menu page $\Menu_0(T)$ with probability at least $\frac{11}{12}$ and sees menu page $\Menu_0(T+1)$ with probability smaller than $\frac{11}{12}$. 
Let $\mech_T$ be the mechanism whose corresponding menus $\{\Menu_T(t)\}_{t=1}^T$ contain only the first $T$ pages of $\mech_0$ (the buyer is shown an empty menu at stage $T+1$). 
The next claim states that the revenue achieved by $\mech_T$ approximates $\opt$ within a constant factor.

\begin{claim}
	\label{cl:opt_survive}
	$\opt = \rev(\mech_0) \le 12 \cdot \rev(\mech_T)$.
\end{claim}
\begin{proof}
%Consider the buyer's behavior by mechanism $\mech_0$. 
Let $\tau$ be a random variable that indicates the menu page $\Menu_0(\tau)$ from which the buyer bought her item ($\tau=0$ if nothing was bought).
%Then $\opt  = \rev(\mech_0, \tau \le T) + \rev(\mech_0, \tau > T)$.
%Consider an alternative behavior of the buyer. 
%At the beginning, with probability $\frac{4}{5}$, the buyer has the same behavior as the original model and with probability $\frac{1}{5}$, she does nothing.
In the case when $\tau\le T$, the revenue of $\mech_T$ is at least as large as the revenue of $\mech_0$ for each valuation profile with $\tau\le T$.
%Whenever $\tau \le T$, the buyer has the same behavior if we show menu $\Menu_T$ instead of $\Menu_0$, since their first $T$ pages are the same. 
%Thus, $\rev(\mech_0, \tau \le T) \le \rev(\mech_T)$.

On the other hand, if $\tau > T$ then the buyer must have seen all the first $T+1$ menu pages, which happens with probability at most $\frac{11}{12}$.
Let us analyze a relaxed version of the optimal mechanism that is allowed to adjust its menu pages at every stage $t>T$ after observing the utility $\util(T)$.
Without loss of generality, we assume that all items in $\cup_{t=1}^{T}\Menu(t)$ get discarded and the relaxed optimum optimizes revenue 
with a smaller supply of the remaining items and worse initial conditions ($\util(T)\ge\util(0)$). Therefore, for each utility level $\util(T)$ the relaxed 
optimal mechanism cannot extract more revenue than $\mech_0$. This implies that the revenue of the optimal mechanism obtained for $\{\vals: \tau> T\}$ is not 
larger than $\frac{11}{12}\cdot\rev(\mech_0)$. Thus $\opt \le \rev(\mech_T) + \frac{11}{12} \cdot \opt$, which concludes the proof.
%
%
%Moreover, since all items are shown at most once, the buyer's behavior after time $t$ is independent of his utility before $t$. Thus, to bound the revenue extracted when $\tau > T$, we consider the following variant of buyer's behavior.
%With probability $\frac{1}{12}$, the buyer does nothing. Otherwise, we shows menu $\Menu_0^{T}$ to the buyer and she continues as the original. Here $\Menu_0^T$ is the menu by deleting the first $T$ pages of $\Menu_0$.
%This establish an upper bound on the revenue,
%\[
%\rev(\mech_0, \tau > T) \le \frac{11}{12} \cdot \rev(\mech_0^T) \le \frac{11}{12} \cdot \opt.
%\]
%Adding the inequalities together, we have $\opt \le \rev(\mech_T) + \frac{11}{12} \cdot \opt$, which concludes the proof after rearranging.
\end{proof}

We collect all the item prices in $\mech_T$, i.e., $\bigcup_{t=1}^{T}\Menu(t)$, and sort them in a decreasing order. Furthermore, we greedily pick the highest prices 
from $\bigcup_{t=1}^{T}\Menu(t)$ into a set $\textsf{TOP}$ while the probability that the buyer would like any $i\in \textsf{TOP}$ ($\exists i \in \textsf{TOP}:\utili \ge 0$) is at most $\frac{1}{12}$.
%\zhihao{We use $i \in \textsf{TOP},$ and $i\in \expensive$ later several times, I add the following sentences to elaborate.}
Here we slight abuse the notation of $\textsf{TOP}$, to denote the corresponding set of items.
This convention is also applied later to a menu of prices.
We put $\textsf{TOP}$ expensive items into a collection $\Mexp$ of menus $\{\Mexp(t)\}_{t=1}^{T}$, where $\Mexp(t)\eqdef\Menu(t)\cap \textsf{TOP}$. 
%In addition, we include all items from $\Menu(T)$ into $\Mexp(T)$.
In addition, the remaining items are placed into the collection $\Mbait$ of menus $\{\Mbait(t)\}_{t=1}^{T}$, where $\Mbait(t)\eqdef\Menu(t)\setminus \textsf{TOP}$. 
The item prices in $\Mexp$ and $\Mbait$ after some modification will serve as expensive and bait items in our bait mechanism. 
%\zhihao{do we?}
%We will refer sometimes to $\Mexp$ and $\Mbait$ as corresponding collections of item prices. 
We denote by $\bar{p}_b\eqdef\max\{p\in\Mbait\}$ and by $\ell\eqdef |\cup_{t}\Mexp(t)|$.
%Consider all menu items on $\Menu_T$ in a decreasing order according to their prices.
%We greedily collect item-price pairs into a new menu $\Mexp$ so that the buyer likes any one of the items in $\Mexp$ with probability at most $\frac{1}{12}$, until an extra item-price pair breaks this property.
%In addition, for each collected item-price pair from the $t$-th page of $\Menu_T$, we also include it into the $t$-th page of $\Mexp$.
%
%We use $\Mbait$ to denote the menu by removing $\Mexp$ from $\Menu_T$ and use $\bar{p}_b$ to denote the highest price in $\Mbait$. 
%Observe that by merging $\Mbait$ and $\Mexp$, we have the same menu as $\Menu_T$.
%Intuitively, $\Mexp$ corresponds to expensive items of $\Menu_T$ while $\Mbait$ corresponds to bait items.
By the definition of $\Mexp$, we have
\begin{equation}
1 - \prod_{p \in \Mexp} F(p) \le \frac{1}{12} < 1 - F(\bar{p}_b)\cdot \prod_{p \in \Mexp} F(p) .
\label{eq:expensive}
\end{equation}
We now bound the revenue of $\mech_T$ by the greedy revenue bound applied to the prices $\Mexp$. 
%We recall that items are shown to the greedy buyer one by one in a decreasing order of prices, and the buyer buys the first item she likes. 
%For notation simplicity, let $\ell = |\Mhigh|$.
%\nick{(i) instead of introducing new variable $l$, we can define sets of high and low prices. (ii) Explicitly say which menu pages we look at in $\mech_1=\mech_T$ 
%(helps to recall the previous step).}

\begin{claim}
	\label{cl:upper_one_time}
	$\rev(\mech_T) \le \myopic{\Mexp} + \bar{p}_b \le 50 \cdot \upricing(\ell)$.
	%\nick{Why 50, and not 26?}\zhihao{a factor of 2 from Myopic to Upricing}
\end{claim}
\begin{proof}
To obtain the first inequality we consider two cases depending on whether the buyer chose item from (i) $\Mexp$, or from (ii) $\Mbait$. We observe that
the expected revenue obtained from items in $\Mexp$ is not more than $\myopic{\Mexp}$ and the expected revenue obtained from the items in $\Mbait$
is not more than $\bar{p}_b$.
%We separate the revenue of $\mech_T$ into two cases, when the buyer eventually takes an item from $\Mexp$ or $\Mbait$.
%Clearly, the revenue extracted from $\Mexp$ is upper bounded by $\myopic{\cup_{t}\Mexp(t)}$.
%Moreover, the revenue extracted from $\Mbait$ is trivially upper bounded by $\bar{p}_b$, since this is the highest price of all items in $\Mbait$.
%This finishes the proof of the first inequality.

To derive the second inequality we use Equation~\eqref{eq:expensive} to obtain
%\begin{align*}
%\frac{\bar{p}_b}{12}   \le \left( 1 -  F(\bar{p}_b)\cdot\prod_{p \in \Mexp} F(p) \right)\cdot \bar{p}_b & \le \myopic{\Mexp \cup \{\bar{p}_b\}} \\
%& \le \myopic{\ell+1} \le 2 \myopic{\ell}.
%\end{align*}
\[
\frac{\bar{p}_b}{12}   \le \left( 1 -  F(\bar{p}_b)\cdot\prod_{p \in \Mexp} F(p) \right)\cdot \bar{p}_b \le \myopic{\Mexp \cup \{\bar{p}_b\}} 
\le \myopic{\ell+1} \le 2 \myopic{\ell}.
\]
Note that the optimal \spm\ (sequential posted pricing for selling one item to many bidders) and \myopicname\ have the same revenue (see Appendix~\ref{app:sec:spm} for more details). We conclude the proof by applying the well-known Fact %~\ref{fact:anonymous_1} 
below.
%Observe the equivalence between \spm\ and \myopicname\, and the equivalence between $\uspm$ and $\upricing$ (see Appendix~\ref{app:sec:spm}), together with Fact~\ref{fact:anonymous} below, we conclude the proof by applying $\myopic{\ell} \le 2 \cdot \upricing(\ell)$.
\begin{fact*}[\cite{ChawlaHK07, DuttingFK16}]
	%\label{fact:anonymous_1}
	$\textsf{SPM} \le 2\cdot\textsf{U-SPM}=2\cdot\upricing.$
	\qedhere
\end{fact*}
%\begin{corollary}
%	\label{cor:myopic_upricing}
%	$\myopic{\ell} \le 2\upricing(\ell)$.
%\end{corollary}
\end{proof}

%If $\Mhigh$ has more menu items than $\Mlow$, we move lowest-priced items back from $\Mhigh$ to $\Mlow$ so that they have the same number of menu items.
%We denote the new menus as $\Moffer$ and $\Mbait$. In case $\Mhigh$ has fewer menu items than $\Mlow$, let $\Moffer=\Mhigh$ and $\Mbait=\Mlow$.
%For notation simplicity, let $\ell = |\Moffer|$.

In the following we first analyze the bait items and show that comparable control over buyer's utilities can be achieved with a collection of simple menus.
Let $u(t)$ be the buyer's utility derived from the menu page $\Mbait(t)$ and $x(t) = u(t)-t\Delta$, for all $t \in [T]$.
\begin{claim}
\label{cl:sq_non_decreasing}
$\{x(t)\}_{t=1}^{T}$ is a non-decreasing sequence with probability at least $\frac{5}{6}$.
\end{claim}
\begin{proof}
If the buyer does not like any item from $\Mexp$, she behaves exactly the same as if she was offered menus $\Mbait$ instead of $\Menu_T$.
%\zhihao{$i \in \Mexp$}
Therefore, $\prob{\text{buyer sees all } \Mbait(t)}\ge \prob{\text{buyer sees all } \Menu_T(t)} - \prob{\exists i\in\Mexp:\utili\ge 0} 
\ge \frac{11}{12}-\frac{1}{12}=\frac{5}{6}.$
%\begin{align*}
%& \prob{\text{buyer sees all} \Mbait(t)} \\
%\ge & \prob{\text{buyer sees all pages of } \Menu_T} - \prob{\text{buyer likes one of the items of } \Mexp} \\
%\ge & \frac{11}{12} - \frac{1}{12} = \frac{5}{6}.
%\end{align*}
As the buyer gets to see all menu pages of $\Mbait$ if and only if $\{x(t)\}_{t=1}^{T}$ is non-decreasing, we conclude the proof.
\end{proof}

In fact, we can have a separation of the supports of random variables $x(t)$'s with only a constant factor loss in probability.
%The next technical lemma shows that with only a constant factor loss in probability, we can have separation of the supports value of all $x(t)$'s.
The following lemma is the central piece of our analysis which will allow us to achieve good control over the buyer's utility $\{\util(t)\}_{t=1}^{T}$.
\begin{lemma}
	\label{lem:utility_control}
Given $n$ independent random variables $\{x_i\}_{i=1}^{n}$.
If $\prob{0 \le x_1 \le x_2 \le \cdots \le x_n} \ge 1 - \eps$, there exist thresholds $0 = \alpha_0 \le \alpha_1 \le \cdots \le \alpha_n < \alpha_{n+1} = \infty$ such that
%\begin{enumerate}
%\item $\prob{\forall 0 \le  i \le \dround{\frac{n-1}{2}}, x_{2i+1} \in [\alpha_{2i}, \alpha_{2i+2}]} \ge 1-2\eps;$
%\item $\prob{\forall 1 \le  i \le \dround{\frac{n}{2}}, x_{2i} \in [\alpha_{2i-1}, \alpha_{2i+1}]} \ge 1-2\eps.$
%\end{enumerate}
\[\prob{\forall i\le\dround{\frac{n-1}{2}}, x_{2i+1} \in [\alpha_{2i}, \alpha_{2i+2}]} \ge 1-2\eps;\quad %\text{and}\quad
\prob{\forall i\le\dround{\frac{n}{2}}, x_{2i} \in [\alpha_{2i-1}, \alpha_{2i+1}]} \ge 1-2\eps.
\]
\end{lemma}

\begin{proof}
Let $\alpha_i$ be the median of $x_i$, i.e., $\prob{x_i \ge \alpha_i} \ge \frac{1}{2}$ and $\prob{x_i \le \alpha_i} \ge \frac{1}{2}$.
We only give the proof to the first statement, as the second one can be derived by the same argument.
If the property does not hold, let $j$ be the smallest index such that $x_{2j+1} \notin [\alpha_{2j}, \alpha_{2j+2}]$. Then either $x_{2j+1} < \alpha_{2j}$, or $x_{2j+1}>\alpha_{2j+2}$. 
Note that the set of random variables $\{x_{2i}\}$ is independent of the choice of $j$ and realization of $\{x_{2i+1}\}$.
In the first case, $x_{2j} \ge \alpha_{2j} > x_{2j+1}$ happens with probability (for a fixed $x_{2j+1}$ and random $x_{2j}$) at least $\frac{1}{2}$. In the second case, $x_{2j+2} \le \alpha_{2j+2} < x_{2j+1}$ happens with probability at least $\frac{1}{2}$. In either case, the monotonicity of $\{x_i\}_{i=1}^{n}$ is violated with probability at least $\frac{1}{2}$. Therefore, 
$
\frac{1}{2} \prob{\exists i, x_{2i+1} \notin [\alpha_{2i}, \alpha_{2i+2}]} \le \prob{\exists i, x_i>x_{i+1}} \le \eps.
$
\end{proof}

We apply this lemma to the above random variables $\{x(t)\}_{t=1}^{T}$, and get a non-decreasing sequence $\{\alpha_t\}_{t=1}^{T}$:
%\begin{enumerate}
	%\item $\prob{\forall 0 \le  t \le \dround{\frac{T-1}{2}}, x(2t+1) \in [\alpha_{2t}, \alpha_{2t+2}]} \ge \frac{2}{3};$
	%\item $\prob{\forall 1 \le  t \le \dround{\frac{T}{2}}, x(2t) \in [\alpha_{2t-1}, \alpha_{2t+1}]} \ge \frac{2}{3}.$
%\end{enumerate}
\[
\prob{\forall t\le\dround{\frac{T-1}{2}}, x(2t+1) \in [\alpha_{2t}, \alpha_{2t+2}]} \ge \frac{2}{3}; \quad%\text{ and }
\prob{\forall t\le\dround{\frac{T}{2}}, x(2t) \in [\alpha_{2t-1}, \alpha_{2t+1}]} \ge \frac{2}{3}.
\]
Let $\Mbaite$ and $\Mbaito$ be the even and odd pages of $\Mbait$ respectively. 
Recall that $\Mbait$ is obtained by removing $\ell$ $\textsf{TOP}$ items from $\Menu_T$. 
Hence, either $\Mbaite$ or $\Mbaito$ has at least $\frac{\ell}{2}$ empty spaces. 
Without loss of generality, we assume $\Mbaite$ has more empty spaces. For 
technical reasons, we remove the last page of $\Mbaite$. Then, $\Mbaite$ has at least $\frac{\ell}{2}-k$ empty spaces on all menu pages.
Let $\lutil[t] = \alpha_{t-1} +t\Delta$ and $\uutil[t] = \alpha_{t+1}+t\Delta$, then%, we have good control on the utilities $u(t)$ of the buyer:
%\begin{enumerate}
	%\item $\prob{\forall 0 \le  t \le \dround{\frac{T-1}{2}}, u(2t+1) \in [\lutil[2t+1], \uutil[2t+1]} \ge \frac{2}{3};$
	%\item $\prob{\forall 1 \le  t \le \dround{\frac{T}{2}}, u(2t) \in [\lutil[2t], \uutil[2t]} \ge \frac{2}{3}.$
%\end{enumerate}
\begin{equation}
\label{eq:even_pages_util}
	\prob{\forall 1 \le  t \le \dround{\frac{T}{2}}, \util(2t) \in [\lutil[2t], \uutil[2t]]} \ge \frac{2}{3}.
\end{equation}

In the remainder of the proof of Theorem~\ref{thm:one_time} we are going to focus only on the pages of $\Mbaite$.
Let $\eT$ be the total number of pages of $\Mbaite$. To simplify notations, we will be using $\util(t)$, $\lutil[t]$, and $\uutil[t]$ to refer to 
the utility derived from the menu page $\Mbaite(t)$ and the corresponding lower and upper bounds.
%To simplify notations, from now on, we redefine $T$ to be the number of pages of $\Mbaite$ 
%and $u(t)$ to be the utility of seeing menu page $\Mbaite(t)$.
%We also re-index $\lutil[t]$'s and $\uutil[t]$'s.
That is, 
\begin{equation}
\prob{\forall t\in[\eT], u(t) \in [\lutil, \uutil]} \ge \frac{2}{3}.
\label{eq:bait_util}
\end{equation}
Observe that by our construction, $\lutil \ge \uutil[t-1] + \Delta$ for all $t\le\eT$.
The Claim~\ref{cl:uutil_T} below shows that the upper bound $\uutil[T_e]$ can be easily recovered by the revenue of a single menu page with $k$ uniformly priced items.
%By definition, we know $\lutil[t] - \uutil[t-1] \ge 2\Delta$ holds for all $t$.
%We first establish a relation between the utility upper bound $\uutil[T]$ and the revenue of uniform pricing over $k$ items. 
%Its proof is deferred to the Appendix~\ref{app:sec:proofs:one}.

\begin{claim}
	\label{cl:uutil_T}
	$\uutil[T_e] + \Delta \le \frac{3}{2} \cdot \upricing(k)$.
\end{claim}
\begin{proof}
Recall that we remove the last page of original $\Mbaite$. Denote the page by $M$.
We know that with probability at least $\frac{2}{3}$, the buyer's utility after seeing menu page $M$ is more than $\uutil[\eT] + \Delta$. 
Consider showing a menu with $k$ items priced at $0$. Note that the utility of seeing this menu stochastically dominates the utility of seeing $M$. Thus, the buyer has utility at least $\uutil[\eT] + \Delta$ with probability at least $\frac{2}{3}$. Finally, consider showing a single page with $k$ items priced at $\uutil[\eT] + \Delta$, we have $\upricing(k) \ge \frac{2}{3} (\uutil[\eT] + \Delta)$.
\end{proof}

The next important step in our analysis is to modify $\Mbaite$ so that on each menu page $\Mbaite(t)$, there are at most two different prices.

\begin{lemma}
	\label{lem:baits_2_prices}
Suppose $\prob{u(t) \in [\lutil[t], \uutil[t]]} = 1-\eps_t$. There exists a menu page $\Menux(t)$ with at most two different prices, % and the same number of items as $\Mbaite(t)$, 
such that $|\Menux(t)|=|\Mbaite(t)|$ and $\prob{\utilx(t) \in [\lutil, \uutil]} \ge 1-2\eps_t$, where $\utilx(t)$ is utility derived from $\Menux(t)$.
\end{lemma}
\begin{proof}
Let $\{p_i\}_{i=1}^{n}$ be all $n$ item prices that appear on the page $\Mbaite(t)$. We know that 
\[
\prob{u(t) \in [\lutil, \uutil]} = \prod_{i=1}^{n} F(p_i+\uutil) - \prod_{i=1}^{n} F(p_i+\lutil) = 1-\eps_t.
\]
Let $a = \prod_{i=1}^{n} F(p_i+\uutil)$ and $b = \prod_{i=1}^{n} F(p_i+\lutil)$. Consider $n$ points $(\ln F(p_i+\uutil), \ln F(p_i+\lutil))$ in $\R^2$. 
By the definition of $a,b$, the center of mass of these $n$ points is $(\frac{\ln a}{n}, \frac{\ln b}{n})$. Since the center of mass must lie inside the convex hull of these points,
there exists a convex combination of just $2$ points that lies below and to the right from the center. In other words, there exists $x\in [0,n]$ and $i_1,i_2 \in [n]$ such that
\[
x \cdot \ln F(p_{i_1}+\uutil) + (n-x) \cdot \ln F(p_{i_2}+\uutil) \ge \ln a \quad\text{and}\quad
x \cdot \ln F(p_{i_1}+\lutil) +(n-x) \cdot \ln F(p_{i_2}+\lutil) \le \ln b.
\]
Without loss of generality, let us assume $p_{i_1} \le p_{i_2}$. We construct a menu page $\Menux(t)$ with $\upround{x}$ items priced at $p_{i_1}$ and $n-\upround{x}$ items priced at $p_{i_2}$.
Then we have
\begin{align*}
\prob{\utilx(t) \in [\lutil, \uutil]} & = F(p_{i_1}+\uutil)^{\upround{x}} F(p_{i_2}+\uutil)^{n-\upround{x}} - F(p_{i_1}+\lutil)^{\upround{x}} F(p_{i_2}+\lutil)^{n-\upround{x}} \\
& \ge F(p_{i_1}+\uutil) \cdot F(p_{i_1}+\uutil)^{x} F(p_{i_2}+\uutil)^{n-x}  - F(p_{i_1}+\lutil)^{x} F(p_{i_2}+\lutil)^{n-x} \\
& \ge (1-\eps_t) \cdot a - b \ge 1-2\eps_t.
\end{align*}
The first inequality follows from the fact that $F(y)\le 1$ for all $y\in\R_{\ge 0}$ and $F(p_{i_1}+\lutil)\le F(p_{i_2}+\lutil)$.
The second inequality follows from the fact that $F(p_{i_1}+\uutil) \ge \prod_{i=1}^{n} F(p_i+\uutil) \ge 1-\eps_t$.
\end{proof}

Finally, we put all the pieces together and prove Theorem~\ref{thm:one_time}. 
Briefly speaking, Lemma~\ref{lem:baits_2_prices} allows us to simplify bait menus with a good control over the buyer's utility and suffer only constant factor losses in the 
success probability and the number of expensive items we could show together with the bait items.

The next natural step is to fill the gaps in the menus of $\{\Menux(t)\}_{t=1}^{\eT}$ with expensive items and hope that the buyer chooses one of them. 
Though the idea is clean, the technical details are involved.
To make the analysis simpler and highlight the structure of the bait mechanism we make a mild assumption on the distribution $\distone$.
%and give readers a quick understanding of how the bait mechanism works. 
The complete proof of the general case is deferred to Section~\ref{app:sec:proofs:one}. We assume $\distone$ is a $(\Delta, \eta)$-spreading distribution (see Definition~\ref{def:spreading} below) and prove that the revenue achieved by a bait mechanism is $O(\frac{1}{\eta})$-approximation to the optimal revenue. 
%Observe that distributions including uniform, exponential distribution, regular distribution etc., admit a large $\eta$ with respect to relatively small $\Delta$.
\begin{definition}[$(\Delta, \eta)$-Spreading]
\label{def:spreading}
A distribution $F$ is a $(\Delta, \eta)$-spreading distribution if $\prob[x\sim F]{x\ge p|x\ge p-\Delta} \ge \eta$ for all $p$ in the support of $\distone$. 
\end{definition}
For example exponential and $\textsf{Uniform}[n]$ distributions are $(\Delta, \eta)$-spreading for small enough $\eta$ and any fixed $\Delta$.
On the other hand, normal distribution is not $(\Delta, \eta)$-spreading for any $\eta,\Delta$. The problem that such distributions pose for our analysis is that
expensive items may interfere with the effect of the bait items causing the buyer to stop her search early, but instead of choosing such expensive item the buyer 
would likely take a cheap bait item.
%Observe that distributions including uniform, exponential distribution, regular distribution etc., admit a large $\eta$ with respect to relatively small $\Delta$.

%\begin{proofof}{Theorem~\ref{thm:one_time}}
\noindent \textbf{Proof of Theorem~\ref{thm:one_time}.\quad}
Let $p^*$ be the optimal price for $\upricing(\frac{\ell}{2})$.
We first consider an easy case when $p^* \le  2\uutil[T_e]$. By Claim~\ref{cl:opt_survive} and~\ref{cl:upper_one_time}, it suffices to give an upper bound on $\upricing(\ell)$. We have
%\begin{align*}
%\upricing(\ell) & \le 2 \cdot \upricing(\frac{\ell}{2}) \le 2 \cdot p^*  \le 2 \cdot \uutil[T] \\
%& \le 3 \cdot \upricing(k) \le 3 \cdot \max_{\mech \in \mathcal{B}} \rev(\mech),
%\end{align*}
\[
\upricing(\ell) \le 2 \cdot \upricing\left(\ell/2\right) \le 2 \cdot p^*  \le 4 \cdot \uutil[\eT]
\le 6 \cdot \upricing(k) \le 6 \cdot \max_{\mech \in \mathcal{B}} \rev(\mech),
\]
where the first inequality follows from Claim~\ref{cl:umenu_half} and the second to the last inequality follows from Claim~\ref{cl:uutil_T}.

Now we assume $p^* >  2\uutil[\eT]$. Let $\oprice \eqdef p^* - \uutil[\eT]$. We consider showing one menu page over $k$ items priced at $\oprice$. If the selling probability $1-F^k(\oprice) \ge \frac{1}{2}$, we have
\[
\upricing(k) \ge (1-F^k(\oprice)) \cdot \oprice \ge \frac{1}{2} \cdot \frac{p^*}{2} = \frac{p^*}{4} \ge \frac{1}{8} \cdot \upricing(\ell).
\]

We assume $1-F^k(\oprice) < \frac{1}{2}$ in the following.
We apply Lemma~\ref{lem:baits_2_prices} to all menu pages of $\{\Mbaite(t)\}_{t=1}^{\eT}$ and denote the new menu as $\{\Menu(t)\}_{t=1}^{\eT}$. 
We fill the $(\frac{\ell}{2}-k)$ gaps in the empty slots of $\Menu$ with expensive items priced at $\oprice$.
%We offer $(\frac{\ell}{2}-k)$ expensive items priced at $(p-\uutil[T])$ in the empty spaces of $\Menu$. 
Then we add an extra menu page with $k$ expensive items priced at $\oprice$ at the end of $\Menu$.
We denote this collection of menus as $\Menu_B$ and the corresponding mechanism as $\mech_B$. Observe that $\Menu_B$ has $(\eT+1)$ menu pages.
The mechanism $\mech_B$ is a bait mechanism with $\Menu$ items being the bait items.
%Considering the menu items of $\Menu$ as bait items, $\mech_B$ falls into the bait mechanisms family $\baitfamily$.

We now show that $\rev(\mech_B) \ge \frac{\eta}{6} \cdot \left( \upricing(\frac{\ell}{2}) - \uutil[\eT] \right)$.
Let $\utilb(t)$ be the buyer's utility derived from the bait item on page $t$, i.e., those on the menu page $\Menu(t)$. Let $\expensive$ be the
set of expensive items, i.e., not bait items, in $\Menu_B$. We will study the revenue of $\mech_B$ only obtained when the following event occurs
$	\eventi[1]\eqdef\left\{\vals:\utilb(t)\in [\lutil, \uutil]~~\forall t \in [\eT]\right\}$.
%and $\eventi[2]\eqdef\left\{\vals: \exists i\in\expensive\text{ s.t. }\vali\ge p^* \right\}$.
%First, we restrict to the revenue when the following events happen:
%We consider the revenue extracted by $\mech_B$ when the following events happen:
%\begin{enumerate}
	%\item[$\bold{E_1}$]: $u'(t) \in [\lutil, \uutil]$ for all $t \in [\eT]$.
	%\item[$\bold{E_2}$]: the buyer values at least one of the $\frac{\ell}{2}$ expensive items more than $p$.
%\end{enumerate}
By Lemma~\ref{lem:baits_2_prices} and Fact~\ref{fa:eps_to_2eps}, we know that
\begin{equation}
\label{eq:prob_e1}
\prob{\bold{E_1}} = \prob{\forall t\in[\eT], \utilb(t) \in [\lutil, \uutil]} \ge \prod_{t\in[\eT]}(1-2\eps_t) \ge \frac{1}{3},
\end{equation}
since $\prod_{t\in[\eT]} (1-\eps_t) \ge \frac{2}{3}$ by Equation~(\ref{eq:bait_util}).

\begin{fact}
	\label{fa:eps_to_2eps}
	$\prod_{t}(1-2\eps_t) \ge 2 \cdot \prod_{t}(1-\eps_t) - 1.$
\end{fact}
\begin{proof}
	Let $h(\vec{\eps}) = \prod_{t=1}^{T} (1-2\eps_t) + 1 - 2 \prod_{t=1}^{T} (1-\eps_t)$.
	Then $\frac{\partial h}{\partial \eps_s} = -2 \prod_{t\ne s} (1-2\eps_t) + 2 \prod_{t\ne s} (1-\eps_t) \ge 0$.
	It follows that the minimum of $h$ is achieved when $\eps_t=0$ for all $t$, in which case $h(\vec{0}) = 0$.
\end{proof}

%The events $\eventi[1]$ and $\eventi[2]$ are independent, since $\eventi[1]$ only depends on the values of the bait items in 
%$\mech_B$ and $\eventi[2]$ depends only on the values of items in $\expensive$. We claim that conditioned on events $\eventi[1]$ 
%and $\eventi[2]$, the buyer purchases an expensive item priced at $\oprice\eqdef(p^*-\uutil[\eT])$ in $\mech_B$ with probability 
%at least $\eta$.
%We delay the proof to the end and provide the revenue guarantee of $\mech_B$.
\begin{claim*}
The revenue of $\mech_B$ conditioned on $\eventi[1]$ is at least $\frac{\eta}{2} \cdot \left( \upricingmech{\frac{\ell}{2}} - \uutil[\eT] \right)$.
%conditioned on events $\eventi[1]$ 
%and $\eventi[2]$, The buyer purchases an expensive item priced at $\oprice\eqdef(p^*-\uutil[\eT])$ in $\mech_B$ with probability 
%at least $\eta$.
\end{claim*}
\begin{proof}
%	\zhihao{again $i \in \expensive_t$}
Let $\expensive_t$ be the expensive items on page page $t$, i.e. $\expensive_t = \Menu_B(t) \cap \expensive$.
We consider the first time $s$ (may not exist) that $\expensive_{s-1}$ interfere with the effect of bait items on page $s$,
i.e., 
\[
\max_{i \in \expensive_{s-1}} (\vali - \oprice) \ge \utilb(s) - \Delta.
\]
First, let us assume that such time $s\in\N$ exists.
Note that, conditioned on $\eventi[1]$, the buyer always continues to the next page if the above event does not happen at stage $s$.
Indeed, by construction $\lutil[s]-\uutil[s-1] \ge \Delta$ and $\utilb(s)-\utilb(s-1) \ge \lutil[s]-\uutil[s-1]$ by the definition of $\eventi[1]$ and 
if $\utilb(s) - \max_{i \in \expensive_{s-1}} (\vali - \oprice) \ge \Delta$ then $\utilb(s)\ge \util(s-1)+\Delta$.

Furthermore, we only consider the case when the buyer does not like any expensive items $\expensive_s$ on page $s\in\N$. This happens with probability at least 
$\frac{1}{2}$, as $1-F^k(\oprice) < \frac{1}{2}$. Denote this event by $\eventi[2]$.
Consequently, the buyer stops at stage $s$ when $\eventi[2]$ happens.
Let $r = |\expensive_{s-1}|$. Then $\val \eqdef \max_{i \in \expensive_{s-1}} \vali$ is drawn according to $F^{r}$.
% (unconditionally of time $s$, $\eventi[1]$, $\eventi[2]$).
Observe that $v$ is independent of $\eventi[1], \eventi[2]$.
By the definition of $(\Delta, \eta)$-spreading distribution, we have
\begin{align*}
\Prlong[v\sim F^{r}]{\val \ge \utilb(s)+\oprice \left| \val \ge \utilb(s)+\oprice-\Delta\vphantom{\sum_i^i}\right.} 
&= \frac{1-F^{r}(\utilb(s)+\oprice)}{1-F^{r}(\utilb(s)+\oprice-\Delta)} \\
&\ge  \frac{1-F(\utilb(s)+\oprice)}{1-F(\utilb(s)+\oprice-\Delta)} \ge \eta.
\end{align*}
That is, conditioned on the buyer stopping at stage $s\in\N$, i.e. $\val - \oprice \ge \utilb(s)-\Delta$, the probability that she buys an expensive item is at least $\eta$.
Then the expected revenue is at least $\prob{\eventi[2]} \cdot \eta \oprice \ge \frac{\eta}{2}\oprice$ for all $s\in\N$.
We are left to give a lower bound on the probability that $s$ exists.
We claim that $s\in\N$ when $\max_{i \in \expensive} \vali \ge p^*$. Indeed, we have, $\utilb(s) + \oprice - \Delta \le \uutil[\eT] + \oprice - \Delta \le p^*$ for all $s$.

We conclude that conditioned on $\eventi[1]$, the revenue is at least 
\[
\prob{\max_{i \in \expensive} v_i \ge p^*} \cdot \frac{\eta}{2} \oprice \ge \frac{\eta}{2} \left( \prob{\max_{i \in \expensive} v_i \ge p^*} \cdot p^* - \uutil[\eT] \right) = \frac{\eta}{2} \cdot \left( \upricingmech{\frac{\ell}{2}} - \uutil[\eT] \right).\qedhere
\]

\end{proof}

%First of all, the buyer would not stop after seeing page $M_B(t+1)$ if her favorite item in page $M_B(t)$ belongs to bait items.
%Since otherwise, the buyer has utility $u'(t)$ after seeing $M_B(t)$, which is at most $\uutil \le \lutil[t+1]-\Delta \le u'(t)-\Delta$ by our assumption.
%However, the buyer's utility after seeing $M_B(t+1)$ is at least $u'(t)$. She would like to see the next page since her utility increases by at least $\Delta$.
%
%Consider the first time $t$ such that the buyer's favorite item in $M_B(t)$ is an offer item.
%This $t$ must exist since otherwise, the buyer has utility $u'(t)$ after seeing page $M_B(t)$ and hence, is hooked until the end of the menu ($u'(t) \ge \lutil \ge \uutil[t-1]+\Delta \ge u'(t-1) + \Delta$ by our assumption).
%On the other hand, when the buyer sees the offer item that she values more than $p$, the item gives her utility $p-(p-\uutil[\eT]-\Delta) \ge u'(\eT)$, which leads to a contradiction.
%
%the buyer will not leave and buy bait items, since the buyer is guaranteed to be hooked as $u'(t) \ge \beta_t \ge \gamma_{t-1} + \Delta \ge u'(t-1) + \Delta$. 

%Observe that events $E_1, E_2$ are independent, 
Overall, we have the following revenue guarantee of $\mech_B$,
\[
\rev(\mech_B) \ge \prob{\eventi[1]} \cdot \rev(\mech_B | \eventi[1]) \ge \frac{\eta}{6} \cdot \left( \upricingmech{\frac{\ell}{2}} - \uutil[\eT] \right).
\]

Combining this with Claim~\ref{cl:umenu_half} and~\ref{cl:uutil_T}, we have
\begin{align*}
\upricing(\ell) & \le 2 \cdot \upricingmech{\frac{\ell}{2}} = 2 \cdot \left( \upricingmech{\frac{\ell}{2}} - \uutil[\eT] + \uutil[\eT] \right) \\
& \le \frac{12}{\eta} \cdot \rev(\mech_B) + 3 \cdot \upricing(k) \le O\left(\frac{1}{\eta}\right) \cdot \max_{\mech \in \mathcal{B}} \rev(\mech). 
\end{align*}

\paragraph{Computations.} Here we discuss how to compute approximately optimal $\mech\in\baitfamily$ via polynomial time dynamic programming (DP). 
%We compute the approximately optimal bait mechanism via dynamic programming.
We recall that in the above construction and analysis of the bait mechanism the buyer's utility from the bait 
items at different menu pages have disjoint supports, i.e., with a constant probability we can restrict all $\utilb(t)$ to lie in the specified disjoint intervals. This crucial 
fact allows us to separate the pricing problem of the bait items into independent problems for each individual menu page. 
Indeed, we only need to care about the upper and lower bounds of the confidence interval $[\lutil, \uutil]$ of $\utilb(t)$ on each menu page. 
Another important feature of the constructed bait mechanism is very limited interaction between the bait and expensive items. Namely, the revenue of the bait
mechanism can be described by a single parameter -- the total number of the available slots left for the expensive items. 

More specifically, our DP works as follows. We dynamically fill a two dimensional array $D\InBrackets{\uutil,\ell} \in [0,1]$, where 
$\uutil$ is the upper confidence bound on $\utilb(t)$ and $\ell$ is the total number of the slots available for the expensive items. The value of $D\InBrackets{\uutil,\ell}$ at time $t$ 
represents the highest possible success probability for the bait items to lead the buyer from stage $1$ to stage $t$ such that $\utilb(t)\le\uutil$ and the total number of 
expensive items slots is $\ell\le t\cdot k$. For each $t$ we can efficiently compute $D\InBrackets{\uutil,\ell}$ by setting $\lutil=\uutil[t-1]+\Delta$ and
using $D\InBrackets{\cdot,\cdot}$ at the previous step $t-1$. Note that in each iteration we only need to search over two different prices and over $k$ possible sizes for 
the bait items on the $t$-th menu page. When we run out of the supply $m$, we choose time $T\le \frac{m}{k}$ and the maximal $\ell$ such that 
$D\InBrackets{\uutil[T], \ell}\ge\frac{1}{3}$. Then using the tables $D\InBrackets{\cdot,\cdot}$ for all $t\le T$ we can recursively find good schedule of bait items
and corresponding confidence intervals $\{[\lutil,\uutil]\}_{t=1}^{T}$ that allows us to show $\ell$ expensive items to the buyer with constant probability.
 
Finally, it is easy to calculate the optimal uniform prices for $\ell$ expensive items and obtain the desired guarantee for the bait mechanism in the case 
when value distribution $F$ is $(\Delta, \eta)$-spreading. In the case, when the distribution $F$ is not $(\Delta, \eta)$-spreading we need to do a little bit more work.
However, the task is not very difficult as from DP computations we know the distribution of $\utilb(t+1)$ and interval $[\lutil[t+1],\uutil[t+1]]$ for each $t\in\N$
and can optimize the uniform price $p^*_t$ independently for each menu page $t$.
%Recall that in the analysis of our bait mechanism, with constant probability, we have a separation on the buyer's utilities of observing bait items, i.e. the support of utilities $\utilb(t)$ are disjoint. Moreover, the crucial parameter for estimating revenue is the number of expensive items we can show.
%
%Following this idea, let $D(\uutil, \ell_t, m_t)$ be the highest probability such that after browsing $t$ menu pages, the buyer has utility at most $\uutil$. At the same time, the menu contains $\ell_t$ empty slots and uses $m_t-\ell_t$ bait items. Moreover, at most two different prices are used on each page. We fill this array by running from $t=1$ while supply lasts, i.e. $m_t \le m$.
%
%For the next $(t+1)$ menu page, for each $\uutil[t+1], \ell_{t+1}, m_{t+1}$, we search over all $\uutil, \ell_t$ and then use $\uutil + \Delta$ as the lower bound of $\utilb(t+1)$. 
%That is, we calculate the highest probability such that $\utilb(t+1)$ falls in $[\uutil+\Delta,\uutil[T+1]]$ using two different prices over $(m_{t+1}-m_t)-(\ell_{t+1}-\ell_{t})$ bait items. 
%In this way, we do not need to carry as a parameter the distribution of $\utilb(t)$.
%This results in a polynomial time algorithm, since the above recursion can be computed in polynomial time and the number of states is also polynomial.
%Finally, we choose parameters so that $\ell_t$ is maximized among all $D(\uutil, \ell_t, m_t) \ge \frac{1}{3}$ .
%To sum up, $\opt \le O(\frac{1}{\eta}) \cdot \max_{\mech \in \mathcal{B}} \rev(\mech)$.
%\end{proofof}
\end{proof}

%\section{Warm up}
%\label{sec:warm_up}
%\input{lifetime}
%\input{example}

%\section{Model with Price Discounts}
%\label{sec:multi_time}
%\input{preliminary_multi}
%
%\subsection{Revenue of $\burr$}
%\label{sec:lower_bound}
%\input{mechanism}
%
%%\section{Optimal dynamic mechanism: Upper Bound}
%\subsection{Revenue of the Optimal Mechanism}
%\label{sec:upper_bound}
%\input{optimal_mech}

%\section{Missing Proofs from Section~\ref{sec:one_time}}
%Missing Proofs for Model Without Price Discounts
%Theorem~\ref{thm:one_time}
\section{Proof of Theorem 3.1 for General Distributions}
\label{app:sec:proofs:one}
%\begin{applemmauutil}
	%$\uutil[\eT] + \Delta \le \frac{3}{2} \cdot \upricing(k)$.
%\end{applemmauutil}
%\begin{proof}
%Recall that we remove the last page of original $\Mbaite$. Denote the page by $M$.
%We know that with probability at least $\frac{2}{3}$, the buyer's utility after seeing menu page $M$ is more than $\uutil[\eT] + \Delta$. 
%Consider showing a menu with $k$ items priced at $0$. Note that the utility of seeing this menu stochastically dominates the utility of seeing $M$. Thus, the buyer has utility at least $\uutil[\eT] + \Delta$ with probability at least $\frac{2}{3}$. Finally, consider showing a single page with $k$ items priced at $\uutil[\eT] + \Delta$, we have $\upricing(k) \ge \frac{2}{3} (\uutil[\eT] + \Delta)$.
%\end{proof}
%
%\begin{applemmamath}
	%$\prod_{t}(1-2\eps_t) \ge 2 \cdot \prod_{t}(1-\eps_t) - 1.$
%\end{applemmamath}
%
%\begin{proof}
	%Let $h(\vec{\eps}) = \prod_{t=1}^{T} (1-2\eps_t) + 1 - 2 \prod_{t=1}^{T} (1-\eps_t)$.
	%Then $\frac{\partial h}{\partial \eps_s} = -2 \prod_{t\ne s} (1-2\eps_t) + 2 \prod_{t\ne s} (1-\eps_t) \ge 0$.
	%It follows that the minimum of $h$ is achieved when $\eps_t=0$ for all $t$, in which case $h(\vec{0}) = 0$.
%\end{proof}

\begin{proof}%{Theorem~\ref{thm:one_time} for General Distributions}
We follow the proof of Theorem~\ref{thm:one_time} as before Definition~\ref{def:spreading}. Let $p^*$ be the optimal price for $\upricing(\frac{\ell}{2})$.
We first consider an easy case when $p^* \le  6 \uutil[\eT]$. By Claim~\ref{cl:opt_survive} and~\ref{cl:upper_one_time}, it suffices to give an upper bound on $\upricing(\ell)$. We have
\[
\upricing(\ell) \le 2 \cdot \upricingmech{\frac{\ell}{2}} \le 2 \cdot p^*  \le 12 \cdot \uutil[\eT] \le 18 \cdot \upricing(k) \le 18 \cdot \max_{\mech \in \baitfamily} \rev(\mech),
\]
where the first inequality follows from Claim~\ref{cl:umenu_half} and the second to the last inequality follows from Claim~\ref{cl:uutil_T}.

Now we assume $p^* > 6 \uutil[\eT]$.
We consider the selling probability of showing one menu page with $k$ items priced at $\frac{p^*}{3}$. If $1-F^k\left(\frac{p^*}{3}\right) \ge \frac{1}{2}$, we have
\[
\upricing(k) \ge \left(1-F^k\left(\frac{p^*}{3}\right)\right) \cdot \frac{p^*}{3} \ge \frac{1}{2} \cdot \frac{p^*}{3} = \frac{p^*}{6} \ge \frac{1}{12} \cdot \upricing(\ell).
\]

We assume $1-F^k\left(\frac{p^*}{3}\right) < \frac{1}{2}$ in the following.
Let $\{\Menu(t)\}_{t=1}^{\eT}$ be the menu pages derived from Lemma~\ref{lem:baits_2_prices} for $\{\Mbaite(t)\}_{t=1}^{\eT}$.
%We apply Lemma~\ref{lem:baits_2_prices} to all menu pages of $\{\Mbaite(t)\}_{t=1}^{\eT}$ and denote the new menu as . 
Let $\utilb(t)$ be the buyer's utility derived from $\Menu(t)$.
For each $t$, let $G_t(\cdot)$ be the cumulative density function of $\utilb(t)$ conditioned on that $\utilb(t) \in [\lutil, \uutil]$. Let $\ell_t$ be the number of empty slots on page $\Menu(t)$.
We first consider the case that $\exists t$, such that $\forall p \in [\frac{p^*}{3}, \frac{p^*}{2}]$,
\begin{equation}
\int_{\lutil}^{\uutil} \left( 1-F^{\ell_{t-1}}(p+u) \right) dG_t(u) \le \frac{1}{2} \cdot \int_{\lutil}^{\uutil}\left( 1-F^{\ell_{t-1}}(p-\Delta+u) \right) dG_t(u).
\label{eq:price_t}
\end{equation}
Let $h(i) = \int_{\lutil}^{\uutil}\left( 1-F^{\ell_{t-1}}(\frac{p^*}{2}-i\Delta+u) \right) dG_t(u)$. The above inequality implies that $h(i) \ge 2 h(i-1)$ for all $i \in [\eT]$, as
$p^*\ge 6\uutil[\eT] \ge 6 \eT\Delta$. Note that $F$ is monotonically increasing. We have
\begin{align*}
1-F^{\ell_{t-1}}\left(\frac{p^*}{3}\right) \ge & 1-F^{\ell_{t-1}}\left(\frac{p^*}{2}-\eT\Delta+\lutil\right) \\
\ge &\int_{\lutil}^{\uutil}\left( 1-F^{\ell_{t-1}}\left(\frac{p^*}{2}-\eT\Delta+u\right) \right) dG_t(u) = h(\eT) \ge 2^{\eT} \cdot h(0) \\
= & 2^{\eT} \cdot \int_{\lutil}^{\uutil}\left( 1-F^{\ell_{t-1}}\left(\frac{p^*}{2}+u\right) \right) dG_t(u) \ge 2^{\eT} \cdot \left(1-F^{\ell_{t-1}}(p^*)\right).
\end{align*}
Let $q_1 = 1-F(p^*)$ and $q_2 = 1-F\left(\frac{p^*}{3}\right)$. Observe that $\frac{q_2}{q_1} \ge \frac{1-F^{\ell_{t-1}}(\frac{p^*}{3})}{1-F^{\ell_{t-1}}(p^*)} \ge 2^{\eT}$.
We consider a single menu page over $k$ items priced at $\frac{p^*}{3}$, then
\begin{align*}
\upricing(k) \ge & (1-(1-q_2)^{k}) \cdot \frac{p^*}{3} \ge (1-(1-2^{\eT}\cdot q_1)^{k}) \cdot \frac{p^*}{3} \\
\ge & (1-(1-q_1)^{2^{\eT} k}) \cdot \frac{p^*}{3} \ge (1-(1-q_1)^{\ell/2}) \cdot \frac{p^*}{3} = \frac{1}{3} \cdot \upricingmech{\frac{\ell}{2}},
\end{align*}
where the last inequality follows from the fact that $\ell/2 \le (\eT+1)k \le 2^{\eT}k$. Thus, $\upricing(\ell) \le 2 \cdot \upricing(\frac{\ell}{2}) \le 6 \cdot \upricing(k)$.

Now, we are left with the case that for all $t\in [\eT]$ there exists a $p_{t-1} \in [\frac{p^*}{3}, \frac{p^*}{2}]$, so that 
\[
\int_{\lutil}^{\uutil}\left( 1-F^{\ell_{t-1}}(p_{t-1}+u) \right) dG_t(u) \ge \frac{1}{2} \cdot \int_{\lutil}^{\uutil}\left( 1-F^{\ell_{t-1}}(p_{t-1}-\Delta+u) \right) dG_t(u).
\]
Let $v$ be the valuation of the buyer's favorite item over $\ell_{t-1}$ items.
The above inequality states that
\begin{equation}
\Prlong[\substack{v\sim F^{\ell_{t-1}}\\ \utilb \sim G_t}]{v-p_{t-1} \ge \utilb(t) \left| v-p_{t-1} \ge \utilb(t)-\Delta\vphantom{\sum_{i}^{j}}\right .} \ge \frac{1}{2}.
\label{eq:cond_prob}
\end{equation}

Then, for each $t\in[\eT]$, we fill the empty slots on $\Menu(t)$ with $\ell_t$ expensive items priced at $p_t$.
Observe that $\sum_{t}\ell_t \ge \frac{\ell}{2}-k$.
%On page $\Menux(t)$, let the corresponding prices be $p_{t}$.
We add an extra menu page with $k$ expensive items priced at $p_{\eT+1}=\frac{p^*}{2}$ at the end of $\Menu$.
%Moreover, we add an extra page with $k$ expensive items with price $p_{T+1}=\frac{p}{2}$ at the end of $\Menux$.
We denote this collection of menus as $\Menu_B$ and the corresponding mechanism as $\mech_B$. Observe that $\Menu_B$ has $(\eT+1)$ menu pages.
The mechanism $\mech_B$ is a bait mechanism with $\Menu$ items being the bait items.
%Denote this mechanism as $\mech_B$. Observe that $\Menu_B$ has $(T+1)$ menu pages.
%The mechanism $\mech_B$ is a bait mechanism, by considering the menu items of $\Menux$ as bait items of $\mech_B$.

Now we establish a lower bound on the revenue extracted by $\mech_B$.
Let $\expensive$ be the set of expensive items, i.e., not bait items, in $\Menu_B$. 
We will study the revenue of $\mech_B$ only obtained when the following event occurs
$\eventi[1]\eqdef\left\{\vals:\utilb(t)\in [\lutil, \uutil]~~\forall t \in [\eT]\right\}$.
%and $\eventi[2]\eqdef\left\{\vals: \exists i\in\expensive\text{ s.t. }\vali\ge p^* \right\}$.

% that $\rev(\mech_B) \ge \frac{\eta}{3} \cdot \left( \umenu(\frac{\ell}{2}) - \uutil[T] \right)$.
By Lemma~\ref{lem:baits_2_prices} and Fact~\ref{fa:eps_to_2eps}, we know that 
\begin{equation}
\label{eq:prob_e1}
\prob{\bold{E_1}} = \prob{\forall t\in[\eT], \utilb(t) \in [\lutil, \uutil]} \ge \prod_{t\in[\eT]}(1-2\eps_t) \ge \frac{1}{3},
\end{equation}
since $\prod_{t\in[\eT]} (1-\eps_t) \ge \frac{2}{3}$ by Equation~(\ref{eq:bait_util}).

%The events $\eventi[1]$ and $\eventi[2]$ are independent, since $\eventi[1]$ only depends on the values of the bait items in 
%$\mech_B$ and $\eventi[2]$ depends only on the values of items in $\expensive$. We claim that conditioned on events $\eventi[1]$ 
%and $\eventi[2]$, the buyer purchases an expensive item in $\mech_B$ with probability 
%at least $\frac{1}{2}$.
%We delay the proof to the end and provide the revenue guarantee of $\mech_B$.
\begin{claim*}
	The revenue of $\mech_B$ conditioned on $\eventi[1]$ is at least $\frac{1}{12} \upricing(\frac{\ell}{2})$.
\end{claim*}

\begin{proof}
Let $\expensive_t$ be the expensive items on page page $t$, i.e. $\expensive_t = \Menu_B(t) \cap \expensive$.
We consider the first time $s$ (might not exist) that $\expensive_{s-1}$ interfere with the effect of bait items on page $s$,
i.e., 
\[
\max_{i \in \expensive_{s-1}} (\vali - p_{s-1}) \ge \utilb(s) - \Delta.
\]
First, let us assume that such time $s\in\N$ exists.
Note that, conditioned on $\eventi[1]$, the buyer always continues to the next page if the above event does not happen at stage $s$.
Indeed, by construction $\lutil[s]-\uutil[s-1] \ge \Delta$ and $\utilb(s)-\utilb(s-1) \ge \lutil[s]-\uutil[s-1]$ by the definition of $\eventi[1]$, and 
if $\utilb(s) - \max_{i \in \expensive_{s-1}} (\vali - p_{s-1}) \ge \Delta$ then $\utilb(s)\ge \util(s-1)+\Delta$.
%if $\utilb(s) - \max_{i \in \expensive_{s-1}} (\vali - \oprice) \ge \Delta$ then $\utilb(s)\ge \util(s-1)+\Delta$.

Furthermore, we only consider the case when the buyer does not like any expensive items $\expensive_s$ on page $s\in\N$. This happens with probability at least 
$\frac{1}{2}$, as $1-F^k(p_{s}) \le 1-F^k(\frac{p^*}{3}) < \frac{1}{2}$. Denote this event by $\eventi[2]$.
Consequently, the buyer stops at stage $s$ when $\eventi[2]$ happens.
Recall that $\ell_{s-1} = |\expensive_{s-1}|$ and that $\val \eqdef \max_{i \in \expensive_{s-1}} \vali$ is drawn according to $F^{\ell_{s-1}}$.
% (unconditionally of time $s$, $\eventi[1]$, $\eventi[2]$).
Observe that $v$ is independent of $\eventi[1], \eventi[2]$.
By Equation \eqref{eq:cond_prob}, we have
\begin{align*}
	\Prlong[v\sim F^{\ell_{s-1}}]{v-p_{s-1} \ge \utilb(s) \left| v-p_{s-1} \ge \utilb(s)-\Delta\vphantom{\sum_i^j}\right .} \ge \frac{1}{2}.
\end{align*}
That is, conditioned on the buyer stopping at time $s$, i.e., $\val - p_{s-1} \ge \utilb(s)-\Delta$, the probability that she buys an expensive item is at least $\frac{1}{2}$. Thus, the expected revenue is at least $\frac{p_{s-1}}{2} \ge \frac{p^*}{6}$. 

%Furthermore, we consider when the buyer does not like any expensive items in $\expensive_s$, which happens with probability at least $\frac{1}{2}$ by $1-F^k(p_{s}) \le 1-F^k(\frac{p^*}{3}) < \frac{1}{2}$. Denote this event by $\eventi[2]$.
%Consequently, the buyer stops at stage $s$ when $\eventi[2]$ happens.
%
%Recall that $\ell_{s-1} = |\expensive_{s-1}|$ and that $\val \eqdef \max_{i \in \expensive_{s-1}} \vali$ is drawn according to $F^{\ell_{s-1}}$.
%% (unconditionally of time $s$, $\eventi[1]$, $\eventi[2]$).
%Observe that $v$ is independent of $\eventi[1], \eventi[2]$.
%By Equation (\ref{eq:cond_prob}), we have

%By the definition of $(\Delta, \eta)$-spreading distribution, we have
%\[
%\Prlong[v\sim F^{r}]{\val \ge \utilb(s)+\oprice \left| \val \ge \utilb(s)+\oprice-\Delta\vphantom{\sum_i^i}\right.} 
%= \frac{1-F^{r}(\utilb(s)+\oprice)}{1-F^{r}(\utilb(s)+\oprice-\Delta)} 
%\ge  \frac{1-F(\utilb(s)+\oprice)}{1-F(\utilb(s)+\oprice-\Delta)} \ge \eta.
%\]
%That is, conditioning on the buyer stops at stage $s$, i.e. $\val - \oprice \ge \utilb(s)-\Delta$, the probability that her favorite item remains the expensive item is at least $\eta$.

To sum up, we have shown that for all $s\in\N$, the expected revenue is at least $\prob{\eventi[2]} \cdot \frac{p^*}{6} \ge \frac{p^*}{12}$.
We are left to lower bound the probability that such $s\in\N$ exists.
We claim that $s\in\N$ when $\max_{i \in \expensive} \vali \ge p^*$. Indeed, we have, $\utilb(s) + \oprice - \Delta \le \uutil[\eT] + \oprice - \Delta \le p^*$ for all $s$.

We conclude that conditioned on $\eventi[1]$, the revenue is at least 
\[
\prob{\max_{i \in \expensive} v_i \ge p^*} \cdot \frac{p^*}{12} = \frac{1}{12} \upricingmech{\frac{\ell}{2}}.\qedhere
\]

\end{proof}

Overall, we have the following revenue guarantee of $\mech_B$,
\[
\rev(\mech_B) \ge \prob{E_1} \cdot \frac{1}{12}\upricingmech{\frac{\ell}{2}} \ge \frac{1}{36} \upricingmech{\frac{\ell}{2}}.
\]

Combining this with Claim~\ref{cl:umenu_half} and~\ref{cl:uutil_T}, we have
\[
\upricing(\ell) \le 2 \cdot \upricingmech{\frac{\ell}{2}} \le 72 \cdot \rev(\mech_B) \le 72 \cdot \max_{\mech \in \mathcal{B}} \rev(\mech). \qedhere
\]
\end{proof}

%\section{Missing Proofs from Section~\ref{sec:multi_time}}
%\section{Missing Proofs for Model With Price Discounts}
%\label{app:sec:proofs:multi}
%\input{appendix_proofs_multi}

\section{Open problems}
\label{sec:open}
We conclude with a few remarks. First, in the choice of our model we specifically looked for as simple 
mathematical formulation as possible. Specifically our i.i.d. assumption, although it might
seem restrictive, actually helps to highlight interesting features and structure of the optimal
pricing for the buyer with search costs while keeping the mechanism design problem still interesting 
and nontrivial. We leave as an open question the extension to non identical prior distribution.
A good starting point would be to investigate the monopoly problem in the static regime, where the seller can select only 
up to $k$ out of $m$ items to display to the buyer. For the dynamic setting, it would be interesting to see if the 
decomposition into ``bait'' and ``expensive'' items still holds and, if it holds, which features of the distributions 
matter for such separation.

Second, our model is unavoidably built on a specific assumption of the buyer search behavior. There could be many reasonable
extensions of the model in the latter regard, e.g., there could be some fixed probability of stopping no 
matter what the buyer's utility increment was, or the buyer's cost $\Delta$ and exploration tolerance parameter $k$ may be 
random variables, or the buyer may be becoming more patient as the search successfully progresses.

Third, the approximation guarantees obtained in our work are rather large and not optimized even within the current analysis.
Maybe we could improve the approximation constant in Theorem~\ref{thm:one_time} to a number below $100$ or even $50$, but using the
current technique it still will be a large constant and probably too far from the true value. 
Thus it would be great to see a different approach and techniques with a better approximation guarantees.

Finally, in many settings the seller actually may observe more about buyer's preferences, than what we described in our model. 
E.g., in almost every online shopping scenario the seller can observe the ``cart'' of the buyer, i.e., the current most favorite 
item of the buyer. This observation may in principal change the seller's algorithm. It would be interesting to see how such extra 
information can affect the seller's pricing policy.

\section*{Acknowledgments}
We sincerely thank Nima Haghpanah for many helpful and productive discussions at all stages of this project. We also thank the audience of
the Dynamic Pricing Workshop at University de Chile for many useful comments.

\bibliographystyle{plain}
\bibliography{../bibs,../ref}

\appendix
\section{Connection with Multi-buyer SPM}
\label{app:sec:spm}
There is a close relationship between $\myopicname$ and the well-known {\bf sequential posted pricing} ($\spm$) mechanism.
In sequential posted pricing mechanism there is $1$ item for sale to $n$ i.i.d. buyers. A $\spm$ is characterized by a price vector $\pricev \in \reals^n$.
The buyers come in a sequence, when the $t$-th comes, we offer a take-it-or-leave-it price $\price_t$. The expected revenue of this mechanism is denoted by $\spm(\pricev)$. 
Let $\spm(n)$ be the optimal revenue one can collect by using sequential posted pricing. We use $\uspm$ if we restrict the posted prices to be the same for all buyers.

It is easy to see that any mechanism for $\myopicname$ induces a mechanism for $\spm$, and vice versa.
\begin{claim}
\label{cl:myopic_spm}
$\myopic{n} = \spm(n)$.
\end{claim} 
\begin{proof}
Let $\pricev$ be the optimal price vector for $\myopic{n}$. We use the same prices in the sequential posted pricing mechanism. For any value profile $\valv \in \reals^n$ items, we map item $j$ in 
the setting with a greedy buyer to the $j$-th buyer's value in the $1$-item-$n$-buyer setting. It is easy to see that the greedy buyer picks item $j$ if and only if the $j$-th buyer wins in the sequential posted pricing mechanism.
The same argument holds reversely, i.e. any sequential posted pricing mechanism also induces a menu for a greedy buyer, from which we conclude the statement.
\end{proof}

Furthermore, the argument also applies if we restrict the posted prices to be a uniform one for both $\myopicname$ and $\spm$.
Observe that with uniform price, the revenue extracted from a greedy buyer is the same as using uniform pricing mechanism. Hence, the optimal revenue of $\uspm$ equals to the optimal revenue of \emph{uniform pricing} mechanism.
\begin{claim}
	\label{cl:upricing_uspm}
	$\upricing(n) = \uspm(n)$.
\end{claim} 
%\begin{proof}
%Let $\price$ be the optimal price $\upricing(n)$. We use the same price in the uniform sequential posted pricing mechanism. For any value profile $\valv \in \reals^n$ items, we map item $j$ in the uniform pricing setting to the $j$-th buyer's value in the $1$-item-$n$-buyer setting. It is easy to see that the buyer in the first setting picks an item if and only if one of the buyers purchases the item in the second setting.
%The same argument hold reversely, i.e. any uniform sequential posted pricing mechanism also induces a uniform pricing for a unit-demand buyer, from which we conclude the statement.
%\end{proof}

%For this reason, we sometimes refer a mechanism in the myopic setting to a multiple-rounds sequential posted pricing mechanism.

%\newpage 

\end{document}